\documentclass[11pt]{article}
\usepackage[utf8]{inputenc}
\usepackage[left=1in, right=1in, bottom=1.15in, top=1.1in]{geometry}

\usepackage{MG}

\title{Quadratic Speedup for Computing Contraction Fixed Points}

\author{
Xi Chen\footnote{Supported by NSF grants IIS-1838154, CCF-2106429 and CCF-2107187.}\\Columbia University\\\url{xichen@cs.columbia.edu}\hspace{-0.2cm}
\and Yuhao Li\footnote{Supported by NSF grants IIS-1838154, CCF-2106429 and CCF-2107187.}\\Columbia University\\\url{yuhaoli@cs.columbia.edu}\hspace{-0.2cm}
\and Mihalis Yannakakis\footnote{Supported by NSF grants CCF-2107187, CCF-2212233, and CCF-2332922.}\\Columbia University\\\url{mihalis@cs.columbia.edu}\vspace{0.1cm}
}

\date{}
\begin{document}
\maketitle
\begin{abstract}

We study the problem of finding an $\eps$-fixed point of a contraction map $f:[0,1]^k\mapsto[0,1]^k$ under both the $\ell_\infty$-norm and the $\ell_1$-norm. 
For both norms, we give an algorithm~with~running time $O(\log^{\lceil k/2\rceil}(1/\eps))$, for any constant $k$. These improve upon the previous best $O(\log^k(1/\eps))$-time algorithm for the $\ell_{\infty}$-norm by Shellman and Sikorski~\cite{SS03}, and the previous best $O(\log^k (1/\eps ))$-time algorithm for the $\ell_{1}$-norm by Fearnley, Gordon, Mehta and Savani~\cite{FGMS20}.
\end{abstract}

\thispagestyle{empty}
\newpage 
\setcounter{page}{1}
\section{Introduction}

A map $f:\calM\mapsto \calM$ defined on a metric space $(\calM,d)$ is called a $(1-\gamma)$-\emph{contraction} map for~some $\gamma\in(0,1]$ if it satisfies 
$d(f(x),f(y))\leq (1-\gamma)\cdot d(x,y)$ for all $x,y\in\calM$. Banach~\cite{Ban1922} proved~in 1922 that every contraction has a \emph{unique} fixed point. 
Notably, his proof is constructive: iteratively applying the contraction map $f$ starting from any point $x\in \calM$ always converges to its unique fixed point.
Banach's foundational result, now known as the Banach fixed point theorem, has profoundly influenced numerous disciplines over the past century, including mathematics~\cite{CL55,nash1956imbedding,gunther1989}, economics and game theory~\cite{Shapley53,stokey1989recursive}, and computer science~\cite{bellman1957markovian,howard1960dynamic,De67}. From an algorithmic perspective, many~fundamental computational problems can be reduced to finding the unique fixed point or a high-precision approximate fixed point, which motivates the study~of~fixed point computation in contractions \cite{STW93,HKS99,SS02,SS03,DP11,DTZ18,FGMS20,CLY25,batzioumonotonecontraction,haslebacher2025query}.

One of the most natural and important metric spaces to study 
the computation of \mbox{contraction} fixed points 
 is the cube $[0,1]^k$ with respect to the $\ell_{\infty}$-norm. 
In particular, {well-studied} games~from formal verification, such as parity games~\cite{EJ91,CJKLS22}, mean-payoff games~\cite{ehrenfeucht1979positional,ZP96} as well as simple stochastic games~\cite{Con92}, can all be  reduced to finding~an $\eps$-fixed point\footnote{A point $x\in \calM$ is said to be an $\eps$-approximate fixed point, or simply an $\eps$-fixed point, of $f$ if $d(x,f(x))\le \eps$. For $[0,1]^k$ with respect to the $\ell_p$-norm, this means an $x\in [0,1]^k$ that satisfies $\|f(x)-x\|_p\le \eps$.} of~a $(1-\gamma)$-contraction $f:[0,1]^k\mapsto [0,1]^k$ under the $\ell_{\infty}$-norm, where the dimension $k$ is polynomial but $\eps$ and $\gamma$ are both inverse exponential in the input size of the game. These games are not only important in verification, but also have an intriguing complexity status: Their decision versions are in $\UP\cap\coUP$ but whether polynomial-time algorithms exist remains long-standing open questions. Indeed, the simplest proof of their $\UP\cap\coUP$-membership is based on the contraction fixed point and crucially utilizes the uniqueness guaranteed by Banach's  theorem~\cite{jurdzinski1998deciding,EY10}. 

More recently, a computational problem named ARRIVAL has attracted considerable attention \cite{dohrau2017arrival,karthik2017did,gartner2018arrival,gartner2021subexponential,haslebacher2025arrival}. ARRIVAL can be reduced to finding an~$\eps$-fixed point of~a~$(1-\gamma)$-contraction $f:[0,1]^k\mapsto [0,1]^k$ with respect to the $\ell_{1}$-norm~\cite{gartner2018arrival,haslebacher2025arrival}, and~with similar dependence of $k$ (polynomial), $\eps$ and $\gamma$ (inverse exponential) on the input size. It also shares a similar complexity status: The decision version of \mbox{ARRIVAL~is} known to be in $\UP\cap\coUP$ (again, one approach for showing this is via the contraction fixed point) but no polynomial-time algorithm is known. \medskip

\noindent\textbf{Black box model.}
In this paper, we aim to extend the algorithmic boundary concerning contraction fixed points over $[0,1]^k$ with respect to both the $\ell_{\infty}$-norm and the $\ell_{1}$-norm. We will work on the \emph{black-box} model, where an algorithm can access the unknown contraction map $f:[0,1]^k\mapsto[0,1]^k$ via a query oracle: given the inputs $k,\eps,\gamma$ and the choice of norm (either $\ell_{\infty}$ or $\ell_{1}$), in each round, the algorithm can send a point $x\in[0,1]^k$ to the oracle and receive $f(x)$ in response. 

We denote the two  problems by $\Contraction_{\infty}(\eps,\gamma,k)$ and $\Contraction_{1}(\eps,\gamma,k)$, corres\-pon\-ding to 
  the $\ell_\infty$-norm and the $\ell_1$-norm, respectively.
While the \emph{query complexity} of an algorithm is about the number of oracle queries it uses in the worst case, the focus of this paper is on the~\emph{time complexity}, i.e., the total computational time needed to find an $\eps$-fixed point, assuming that the oracle answers each query in one time step. %
Clearly, the time complexity is always at least as large as the query complexity, but not necessarily vice versa.\medskip

\noindent \textbf{Prior work on contraction fixed point computation.} While Banach's value iteration method works for any metric space $(\calM,d)$, it requires $O((1/\gamma)\log(1/\eps))$ time to converge to an $\eps$-fixed~point which leads to exponential-time algorithms for all applications discussed above. 

The first algorithms designed specifically for $\Contraction_{\infty}(\eps,\gamma,k)$ over $[0,1]^k$  under the $\ell_\infty$-norm were by Shellman and Sikorski. They gave an $O(\log (1/\eps))$-time algorithm when $k=2$~\cite{SS02}  (note~that a simple binary search can solve $\Contraction_{\infty}(\eps,\gamma,1)$ in $O(\log (1/\eps))$ time)~and later an $O(\log^k(1/\eps))$-time algorithm for any constant $k$ \cite{SS03}. For $\ell_1$-norm, Fearnley, \mbox{Gordon,~Mehta} and Savani \cite{FGMS20} gave an algorithm for $\Contraction_{1}(\eps,\gamma,k)$ with a similar $O(\log^k(1/\eps))$ {running time for any constant $k$}. 

We note that these time complexity upper bounds are much better than that of Banach's value iteration method when $k$ is a constant and 
the other two main parameters $\gamma$ and $\eps$ are polynomially related (or even when their logarithms are polynomially related); they remain the best upper bounds for this setting before our work. 
We also note that these upper bounds 
do not \mbox{depend} on $\gamma$; indeed, algorithms of \cite{SS02,SS03,FGMS20} can solve the more general problems of $\NonExp_{\infty}(\eps,k)$, the problem of finding an $\eps$-fixed point in a nonexpansive\footnote{A map $f:\calM\rightarrow \calM$ is said to be nonexpansive if $d(f(x),f(y))\le d(x,y)$ for all $x,y\in\calM$.} map over $[0,1]^k$ under the $\ell_{\infty}$-norm, and $\NonExp_{1}(\eps,k)$, the same problem under the $\ell_{1}$-norm (see \Cref{def:nonexpansive}).

Before presenting the main contributions of the paper, which is on time complexity, we notice that  
  significant progress has been made 
  on the \emph{query} complexities~of
  $\NonExp$ and $\Contraction$.
In \cite{CLY25} Chen, Li and Yannakakis gave an algorithm for $\NonExp_\infty(k,\eps)$ with query complexity $O(k\log(1/\eps))$. Later Haslebacher, Lill, Schnider and  Weber \cite{haslebacher2025query} obtained an algorithm~for $\Contraction_{p}(\eps,\gamma,k)$ for any $\ell_p$-norm of $p\ge 1$ with  $O(k^2\log(1/(\eps\gamma)))$ queries.
However, these algorithms use brute force $\Omega((1/\eps)^k)$-time to search for the next point to query in  each round.\medskip

\noindent\textbf{Our contributions.} 
Our main results are a quadratic speedup over algorithms of \cite{SS03,FGMS20} for $\NonExp_\infty(\eps,k)$ and $\NonExp_1(\eps,k)$, detailed as follows: 
\begin{restatable}{theorem}{theoremlinfty}\label{theorem: linfty}
	For any constant $k$, there is an $O(\log^{\lceil k/2\rceil}(1/\eps))$-time algorithm for $\NonExp_{\infty}(\eps,k)$.
\end{restatable}\vspace{-0.3cm}
\begin{restatable}{theorem}{theoremlone}\label{theorem: l1}
	For any constant $k$, there is an $O(\log^{\lceil k/2\rceil}(1/\eps))$-time algorithm for $\NonExp_1(\eps,k)$.
\end{restatable}

The constants hidden in the two theorems are $c^k$ for some constant $c$ and $k^k$, respectively.
Below we give a high-level overview of our algorithms behind \Cref{theorem: linfty} and \Cref{theorem: l1}.

\subsection{Technical Overview}\label{subsection:overview}

The proof of \Cref{theorem: linfty} starts by introducing a generalization of problem $\NonExp_\infty(\eps,k)$,~called $\NonExp_\infty^\dagger(\eps,k)$, where the unknown nonexpansive map $f$ is from $[0,1]^k$ to $[-\eps,1+\eps]^k$ (instead of $[0,1]^k$). One can still show that $f$ must have an $\eps$-fixed point $x\in [0,1]^k$ satisfying $\|f(x)-x\|_\infty\le \eps$, and the goal of the problem is to find such an $\eps$-fixed point of $f$. (See \Cref{def:nonexpansivegeneralization} for the formal definition of $\NonExp_\infty^\dagger(\eps,k)$.)
Given that $\NonExp_\infty(\eps,k)$ is a special case of $\NonExp_\infty^\dagger(\eps,k)$, it suffices to give an algorithm for $\NonExp_\infty^\dagger(\eps,k)$ with running time $O(\log^{\lceil k/2\rceil} (1/\eps))$.

The proof of \Cref{theorem: linfty}  proceeds in two steps.
First we give an $O(\log(1/\eps))$-time algorithm~for $\NonExp_\infty^\dagger(\eps,k)$ when $k=1$ or $2$. The case when $k=1$ can be solved by a simple binary search; the case when $k=2$, on the other hand, is more involved and is done by a direct reduction to the original problem
  $\NonExp_\infty(\eps,2)$ which is known to have  an $O(\log(1/\eps))$-time algorithm \cite{SS02}.
  
Next we prove a new \emph{decomposition theorem} for $\NonExp_\infty^\dagger(\eps,k)$ (\autoref{thm: linfty decomposition}): If $\NonExp_\infty^\dagger(\eps,a)$ and $\NonExp_\infty^\dagger(\eps,b)$ can be solved in time $t(\eps,a)$ and $t(\eps,b)$, respectively, then $\NonExp_\infty^\dagger(\eps,a+b)$ can be solved in time $O(t(\eps,a)\cdot t(\eps,b))$. 
\Cref{theorem: linfty} then follows by combining the two steps.

To the best of our knowledge, this is the first decomposition theorem for contraction fixed points; previously such decomposition theorems were only known for Tarski fixed points in monotone maps~\cite{FPS22,CL22}.
At a high level, given oracle access to a nonexpansive $f:[0,1]^{a+b}\rightarrow [-\eps,1+\eps]^{a+b}$, the~algorithm (\autoref{alg: linfty}) works by running a $t(\eps,b)$-time algorithm $\calA_2$ on a nonexpansive map $g:[0,1]^b\rightarrow [-\eps,1+\eps]^b$ that is built adaptively query by query: For each point $q^t\in [0,1]^b$ that $\calA_2$ would like to query in round $t$, \autoref{alg: linfty} needs to work on $f$ and return an answer $v^t\in [-\eps,1+\eps]^b$ as the result of $g(q^t)$. We would like the following three conditions to hold:
\begin{flushleft}\begin{enumerate}
\item All results $(q^1,v^1),\ldots,(q^t,v^t)$ about $g$ returned by \autoref{alg: linfty} to $\calA_2$ are consistent with $g$ being a nonexpansive map from $[0,1]^b$ to $[-\eps,1+\eps]^b$. (This guarantees that $\calA_2$ terminates within $t(\eps,b)$ time and in particular, $t(\eps,b)$ rounds, with an $\eps$-fixed point of $g$.)
\item When $\calA_2$ terminates in round $t$ with $q^t$ being an $\eps$-fixed point of $g$, i.e, $\|v^t-q^t\|_\infty\le \eps$, we can use it to obtain an $\eps$-fixed point of the original map $f$ efficiently in time. 
\item For each round, given the query  $q^t$ from $\calA_2$, the time spent by \autoref{alg: linfty} to compute its response $v^t$ to $\calA_2$ should take no more than $t(\eps,a)$ time (i.e., the time needed to run the best algorithm to solved a $\NonExp_\infty^\dagger(\eps,a)$ instance).
\end{enumerate}\end{flushleft}
The challenge behind the design of \autoref{alg: linfty}  lies in satisfying all three conditions simultaneously.
In particular, to return a $v^t$ that satisfies the first condition, \autoref{alg: linfty} cannot just run an $t(\eps,a)$-time algorithm $\calA_1$ to find an $\eps$-fixed point in $f$ restricted to the $[0,1]^a$ cube by setting the last $b$ coordinates to $q^t$. Instead, \autoref{alg: linfty} works by focusing on a carefully constructed $a$-dimensional cube defined by $p_{\min}^t$ and $p_{\max}^t$ (see \autoref{alg: linfty} for details). The use of the slightly harder problem $\NonExp_\infty^\dagger$ (instead of $\NonExp_\infty$) is also crucial to the success of the decomposition theorem; in spirit this is similar to situations where certain inductions work only when one starts with a slightly stronger basis and inductive hypothesis.

The proof of \Cref{theorem: l1} follows a similar strategy, namely, solve low dimensions directly and then use recursion for higher dimensions. 
However, there are key differences in both steps due~to  fundamental differences between the $\ell_\infty$ and $\ell_1$ norms.
In particular, our goal is to give an algorithm for $\Contraction_{1}(\eps,\gamma,k)$ (\autoref{algorithm: time}) with running time $O(k\log^{\lceil k/2\rceil}(1/\eps\gamma))$; \Cref{theorem: l1} follows from a reduction from $\NonExp_1(\eps,k)$ to $\Contraction_{1}(\eps/2,\eps/(2k),k)$ (see Observation~\ref{observation: l1 nonexp to contraction}).  

When $k=1$, the same simple binary search still works (indeed, $\ell_p$ behaves exactly the same for all $p\in[1,\infty)\cup \set{\infty}$ when $k=1$). For the case $k=2$, our $O(\log(1/\eps\gamma))$-time algorithm at a high level works as follows: Suppose that we query the middle point $x=(1/2,1/2)$ and get $f(x)$ such that $\norm{f(x)-x}_1>\eps$. For simplicity, assume $f(x)_1\geq x_1$ and $f(x)_2\geq x_2$. Then we can infer that $f(y)\neq y$ for all $y\in[0,1]^2$ such that $\norm{x-y}_1\leq \norm{f(x)-y}_1$, which is because $\norm{f(x)-f(y)}_1\leq (1-\gamma)\norm{x-y}_1<\norm{f(x)-y}_1$. The key idea then is to analyze the set of such $y$ by looking at a ``dominating coordinate,'' namely, a coordinate $i\in\set{1,2}$ such that $|f(x)_i-x_i|\geq |f(x)_{3-i}-x_{3-i}|$. Suppose that $f(x)_1-x_1\geq f(x)_2-x_2$, then we claim that $f(y)\neq y$ for all $y\in[0,1/2]\times[0,1]$. This is because $\norm{f(x)-y}_1-\norm{x-y}_1\geq (f(x)_1-x_1)-|f(x)_2-x_2|\geq 0$. Intuitively, this means that we can use one query to eliminate the search space by half, and the remaining space remains a nice rectangle so that we know the next query point should be the center of the rectangle. With more details being worked out (such as how to use the condition $\norm{f(x)-x}_1>\eps$), we obtain an $O(\log(1/\eps\gamma))$-time algorithm for $\Contraction_1(\eps,\gamma,2)$.

For higher dimensions $k>2$, it is not clear whether the decomposition approach can be applied due to some technical challenges. Instead, we directly generalize our approach above, which analyzes the remaining search space by identifying a ``dominating coordinate.'' Here, a strawman definition of a dominating coordinate would be an $i$ such that $|f(x)_i-x_i|\geq \sum_{j\neq i}|f(x)_j-x_j|$. However, it is immediately clear that in higher dimensions, such an $i$ might not exist. It turns out that for the approach to work, it suffices to find a \emph{weakly-dominating coordinate}, defined as an $i$ such that $|f(x)_i-x_i|\geq \sum_{j\neq i}|f(x)_j-x_j|-\eps\gamma/4$. This is where the recursion comes in. We fix the first two coordinates, and work on a $k-2$ dimensional rectangle to recursively find an $\eps\gamma/4$-fixed point with respect to the remaining $k-2$ coordinates. Then either $i=1$ or $i=2$ will be a weakly-dominating coordinate. Note that the boost on the approximation precision $\eps\gamma/4$ instead of $\eps$ in recursion is important, and this is why the constants hidden in \Cref{theorem: l1} are $k^k$.

\subsection{Other Related Work}

The problem of
finding an approximate fixed point in a contraction map under the $\ell_2$-norm has~long been~known to admit polynomial-time algorithms~\cite{STW93, HKS99}. The best time bounds for other $\ell_p$-norms (such as $\ell_\infty$ and $\ell_1$ studied in this paper), however, remain exponential in $k$. %

As natural total search problems lacking polynomial-time algorithms, the computational problem of computing contraction fixed points naturally fits in the landscape of $\TFNP$~\cite{MegiddoPapadimitriou}. In fact, it was one of the motivating problems behind the introduction of the class $\CLS$~\cite{DP11}. Recently, to better capture the complexity of problems with intrinsically unique solutions, \cite{FGMS20} defined the subclass $\UEOPL$ and proved (among other problems) the $\UEOPL$-membership of computing the \emph{exact} fixed point of piecewise-linear contraction maps, an easier white-box variant of the black-box problem we consider here. Very recently, Batziou et al. \cite{batzioumonotonecontraction} showed that, \emph{in the black-box model}, finding an \emph{approximate} fixed point for maps that are both monotone and contracting under the $\ell_\infty$-norm is in $\UEOPL$. They also provided an algorithm with time complexity $O(\log^{\lceil k/3\rceil}(1/\eps))$, which is faster than the best-known algorithms for maps that are just monotone \cite{CL22} or just contracting (\cite{SS03} and the present paper).

Despite considerable effort in the $\TFNP$ characterization of contraction fixed points under $\ell_p$-norms (e.g. see \cite{EY10,DP11,DTZ18,FGMS20,Hol21,FGHS23}), it remains unclear whether they are hard for any $\TFNP$ class studied in the literature. As mentioned earlier, significant~progress~has been made  on the query complexity of contraction fixed points \cite{CLY25,haslebacher2025query} which in particular ruled outs black-box reductions from known $\TFNP$ classes to contraction fixed points under $\ell_p$-norms; we refer interested readers to the discussion in \cite{CLY25} for further details.
\section{Preliminaries}

We start with the definitions of contraction and nonexpansive maps.

\begin{definition}[Contraction]
Let $0<\gamma<1$ and $(\calM,d)$ be a metric space. A map $f:\calM\mapsto\calM$ is a $(1-\gamma)$-\emph{contraction map} with respect to $(\calM,d)$ if $d(f(x),f(y))\leq (1-\gamma)\cdot d(x,y)$ for all $x,y\in \calM$.

A map $f:\calM\mapsto\calM$ is said to be \emph{nonexpansive} if 
  $d(f(x),f(y))\leq d(x,y)$ for all $x,y\in \calM$.
\end{definition}

Every contraction map has a unique fixed point, i.e. $x^*$ with $f(x^*)=x^*$ by Banach's fixed point theorem. Every nonexpansive map over $[0,1]^k$ with respect to the $\ell_p$-norm, $p\in[1,\infty)\cup\set{\infty}$, has a (not necessarily unique) fixed point by Brouwer's fixed point theorem.
In this paper, we study the time complexity of finding an $\eps$-fixed point of a $(1-\gamma)$-contraction map $f$  over the $k$-dimensional cube $[0,1]^k$ with respect to the $\ell_p$-norm, where $p\in\set{1,\infty}$:

\begin{definition}[$\Contraction_p(\eps,\gamma,k)$]
In problem $\Contraction_p(\eps,\gamma,k)$, we are given 
  oracle access to a $(1-\gamma)$-contraction map $f$ over $[0,1]^k$ with respect to the $\ell_p$-norm, i.e., $f$ satisfies
$$
\norm{f(x)-f(y)}_p\le (1-\gamma)\cdot \norm{x-y}_p,\quad \text{for all $x,y\in [0,1]^k$}.
$$
The goal is to find an $\eps$-fixed point of $f$, i.e., a point $x \in [0,1]^k$ such that $\norm{f(x )-x}_p\leq \eps$. 
\end{definition}

We define similar problems for nonexpansive maps over $[0,1]^k$ with respect to the $\ell_p$-norm:

\begin{definition}[$\NonExp_p(\eps,k)$]\label{def:nonexpansive}
In problem $\NonExp_p(\eps,k)$, we are given 
  oracle access to a nonexpansive map $f :[0,1]^k\rightarrow [0,1]^k$ with respect to the $\ell_p$-norm, i.e., $f$ satisfies
$$
\norm{f(x)-f(y)}_p\le \norm{x-y}_p,\quad \text{for all $x,y\in [0,1]^k$}.
$$
The goal is to find an $\eps$-fixed point of $f$, i.e., a point $x\in [0,1]^k$ such that $\|f(x)-x\|_p\le \eps$. 
\end{definition}

We will use the following notation. 
Given a positive integer $m$, we write $[m]$ to denote the set $\{1,\ldots,m\}$. Given two vectors $\alpha,\beta\in\mathbb{R}^m$, we say $\alpha\preceq \beta$ if $\alpha_i\leq \beta_i$ for all $i\in[m]$. Given $\alpha,\beta\in\mathbb{R}^m$ such that $\alpha\preceq \beta$, we use 
$\calB(\alpha,\beta)\coloneqq \set{x: \alpha\preceq x \preceq \beta}$ to denote the \emph{box}  defined by $\alpha$ and $\beta$. Given $\alpha,\beta\in\mathbb{R}^m$ such that $\alpha\preceq \beta$ and $\eps>0$, we 
write $\calB(\alpha,\beta,\pm\eps)$ to denote 
$$\calB(\alpha,\beta,\pm \eps)\coloneqq \big\{x\in \mathbb{R}^m: (\alpha_1-\eps,\cdots,\alpha_m-\eps)\preceq x \preceq (\beta_1+\eps,\cdots,\beta_m+\eps)\big\}.$$

Our main algorithm in \Cref{sec:alg1} for
  $\NonExp_\infty(\eps,k)$ will work on the following generalization:

  \begin{definition}[$\NonExp_\infty^\dagger(\eps,k)$]\label{def:nonexpansivegeneralization}
In $\NonExp_\infty^\dagger(\eps,k)$ we are given $\alpha_{\min},\alpha_{\max}\in [0,1]^k$~with $\alpha_{\min}\preceq \alpha_{\max}$ and oracle access to a nonexpansive map $f:\calB(\alpha_{\min},\alpha_{\max})\rightarrow \calB(\alpha_{\min},\alpha_{\max},\pm \eps)$
with respect to the $\ell_\infty$-norm. The goal is to find an $\eps$-fixed point of $f$, i.e., a point $x\in \calB(\alpha_{\min},\alpha_{\max})$ such~that $\|f(x)-x\|_\infty\le \eps$.
Note that such a point always exists. To see this, we write  $g:\calB(\alpha_{\min},\alpha_{\max})\rightarrow \calB(\alpha_{\min},\alpha_{\max})$ to denote the map obtained by truncating $f$ as follows:
$$
g(x)_i:=\max\big(\min(f(x)_i,\alpha_{\max,i}),\alpha_{\min,i}\big),\quad\text{for every $i\in [k]$.}
$$
It is easy to verify that $g$ remains nonexpansive and thus, has a fixed point $x$ by Brouwer's fixed point theorem.
Any fixed point of $g$ must be an $\eps$-fixed point of $f$.
\end{definition}

\section{Algorithm for $\NonExp_{\infty}(\eps,k)$}\label{sec:alg1}

To prove \autoref{theorem: linfty}, it suffices to give an algorithm that solves the generalization $\NonExp_\infty^\dagger(\eps,k)$ of $\NonExp_\infty(\eps,k)$ in 
  $O(\log^{\lceil k/2\rceil}(1/\eps))$ time.
To this end, we start by noting that when $k\in \{1,2\}$~it follows from Shellman~and Sikorski~\cite{SS02} that $\NonExp_\infty^\dagger(\eps,k)$ can be solved in $O(\log(1/\eps))$ time:

\begin{lemma}\label{lem:ss02}
	There is an $O(\log(1/\eps))$-time algorithm for $\NonExp^\dagger_{\infty}(\eps,k)$ when $k\in \{1,2\}$.
\end{lemma}
\begin{proof}
First the problem $\NonExp_\infty^\dagger(\eps,1)$ can be easily solved by binary search in time $O(\log(1/\eps))$ so we focus on $\NonExp_\infty^\dagger(\eps,2)$ below.

Suppose that we are given $\alpha_{\min},\alpha_{\max}\in [0,1]^2$~with $\alpha_{\min}\preceq \alpha_{\max}$ and oracle access to a~nonexpansive map $f:\calB(\alpha_{\min},\alpha_{\max})\rightarrow \calB(\alpha_{\min},\alpha_{\max},\pm \eps)$
under the \mbox{$\ell_\infty$-norm.} We define  $g:\calB(\alpha_{\min},\alpha_{\max})$ $\rightarrow \calB(\alpha_{\min},\alpha_{\max})$ as the truncation of $f$, the same as above:
$$
g(x)_i:=\max\big(\min(f(x)_i,\alpha_{\max,i}),\alpha_{\min,i}\big),\quad\text{for both $i\in \{1,2\}$.}
$$
Given that $g$ is also nonexpansive, we can use the algorithm of Shellman~and Sikorski~\cite{SS02} to find an $\eps$-fixed point of $g$ in time $O(\log(1/\eps))$. Denote such a point by $x^*\in\calB(\alpha_{\min},\alpha_{\max})$.

If $\norm{f(x^*)-x^*}_{\infty}\leq\eps$, then we are done. Otherwise, it must be the case that at least one coordinate $i\in\set{1,2}$ of $f(x^*)$ gets truncated, and $|f(x^*)_i-x^*_i|>\eps$. Without loss of generality, assume that $f(x^*)_i\geq x^*_i$ for both $i\in\set{1,2}$ and we
 consider two cases: (1) $f(x^*)_i\not\in[\alpha_{\min,i},\alpha_{\max,i}]$ for both $i\in\set{1,2}$, and (2) there is exactly one coordinate $i\in\set{1,2}$ such that $f(x^*)_i\not\in[\alpha_{\min,i},\alpha_{\max,i}]$.
\begin{flushleft}\begin{enumerate}
\item For the first case, we know that $f(x^*)>\alpha_{\max,i}$ for both $i\in\set{1,2}$. We show that $\alpha_{\max}$\\ must be a desired $\eps$-fixed point of $f$. On the one hand, since $f(\alpha_{\max})\in\calB(\alpha_{\min},\alpha_{\max},\pm \eps)$, we know that $f(\alpha_{\max})_i-\alpha_{\max,i}\leq \eps$ for both $i\in\set{1,2}$. On the other hand, given that $\norm{x^*-\alpha_{\max}}_{\infty}\leq\eps$, we have  $\norm{f(\alpha_{\max})-f(x^*)}_{\infty}\leq\eps$ and thus, $$f(\alpha_{\max})_i\geq f(x^*)_i-\eps>\alpha_{\max,i}-\eps$$ for both $i\in\set{1,2}$. This shows that $\|f(\alpha_{\max})-\alpha_{\max}\|_\infty\le \eps$.

\item The second case is more involved. Assume without loss of generality that $f(x^*)_1\leq \alpha_{\max,1}$ and $f(x^*)_2> \alpha_{\max,2}$.  We have that $f(x^*)_2>x^*_2+\eps>\alpha_{\max,2}$.  Consider the point $y$ with $y_1=x^*_1$ and $y_2=\alpha_{\max,2}$. So $\norm{y-x^*}_{\infty}=\alpha_{\max,2}-x^*_2<\eps$ and thus, $\|f(y)-f(x^*)\|_\infty<\eps$. From this we have $|f(y)_2-y_2|\leq \eps$:
  $f(y)_2-y_2\le \eps$ follows from $f(y)_2\le \alpha_{\max,2}+\eps=y_2+\eps$;
  $f(y)_2\ge y_2-\eps$ because otherwise, 
  $\|f(y)-f(x^*)\|_\infty\ge |f(y)_2-f(x^*)_2|> \eps$, a contradiction.

Furthermore, we have $f(y)_1> f(x^*)_1-\eps\geq x^*_1-\eps= y_1-\eps$. If $f(y)_1\leq y_1+\eps$, then $y$ is a desired $\eps$-fixed point of $f$ and we are done. Otherwise, we have $f(y)_1>y_1+\eps$ and we show below that $y'$, defined as $y'_2\coloneqq y_2=\alpha_{\max,2}$ and $$y'_1\coloneqq \min(\alpha_{\max,1},y_1+\eps)=
\min(\alpha_{\max,1},x^*_1+\eps),$$ must be a desired $\eps$-fixed point of $f$. 

First note that $\norm{y'-x^*}_{\infty}, \|y'-y\|_\infty\leq\eps$ and thus,
$\|f(y')-f(x^*)\|_\infty, \|f(y')-f(y)\|_\infty\le \eps$.
It follows from a similar argument that $|f(y')_2-y'_2|\leq\eps$. 
We also have 
$$f(y')_1\geq f(y)_1-\eps>y_1\geq y'_1-\eps.$$
Finally we have $f(y')_1\leq y'_1+\eps$ because: if $y'_1=\alpha_{\max,1}$, then $f(y')_1\leq \alpha_{\max,1}+\eps=y'_1+\eps$; if $y'_1<\alpha_{\max,1}$, then $f(y')_1-f(x^*)_1\leq\eps$, which implies that $f(y')_1\leq f(x^*)_1+\eps\leq y'_1+\eps$.  
\end{enumerate}\end{flushleft}
This finishes the proof of the lemma.
\end{proof}

\autoref{theorem: linfty} then follows by combining \Cref{lem:ss02} with the following decomposition theorem:%

\begin{theorem}\label{thm: linfty decomposition}
	Suppose that $\NonExp^\dagger_{\infty}(\eps,a)$ can be solved in $q(\eps,a)$ queries and $t(\eps,a)$ time, and $\NonExp^\dagger_{\infty}(\eps,b)$ can be solved in $q(\eps,b)$ queries and $t(\eps,b)$ time. Then $\NonExp^\dagger_{\infty}(\eps,a+b)$ can be solved in $q(\eps,a)\cdot q(\eps,b)$ queries and $O((q(\eps,b)+t(\eps,a))\cdot t(\eps,b))$ time.
\end{theorem}

One may notice an overhead term $q(\eps, b)$ in the time complexity, due to line~\ref{line: linfty construct} of \autoref{alg: linfty}. This is not an issue, as when we apply the decomposition theorem to obtain \autoref{theorem: linfty}, we can choose $b\leq a$ so that the $q(\eps,b)$ can be absorbed into the $t(\eps,a)$ term.

\begin{algorithm}[!t]
	\caption{Algorithm for $\NonExp^\dagger_{\infty}(\eps,{a+b} )$ via $\NonExp^\dagger_{\infty}(\eps,a)$ and 
	  $\NonExp^\dagger_{\infty}(\eps,b)$}
	\label{alg: linfty}
	\KwIn{An $\eps\in(0,1]$, $\alpha_{\min},\alpha_{\max}\in[0,1]^a$, $\beta_{\min},\beta_{\max}\in[0,1]^b$ such that $\alpha_{\min}\preceq \alpha_{\max}$ and $\beta_{\min}\preceq \beta_{\max}$, and an oracle access to a nonexpansive map
		$f: \calB((\alpha_{\min},\beta_{\min}),(\alpha_{\max},\beta_{\max}))\to \calB((\alpha_{\min},\beta_{\min}),(\alpha_{\max},\beta_{\max}),\pm \eps)$.
	}
	\KwOut{A point $x^*\in \calB((\alpha_{\min},\beta_{\min}),(\alpha_{\max},\beta_{\max}))$ such that $\norm{f(x^*)-x^*}_\infty\leq \eps$.
	}
	
	\BlankLine
	
	Let $\calA_1$ be an algorithm for $\NonExp^\dagger_{\infty}(\eps,a)$ and $\calA_2$ be an algorithm for $\NonExp^\dagger_{\infty}(\eps,b)$.\\

	\For{$t=1,2,\ldots$}{

	For each previous round $r\in[t-1]$, let $q^r\in \calB(\beta_{\min},\beta_{\max})$ be the point queried by $\calA_2$ and $v^r\in \calB(\beta_{\min},\beta_{\max},\pm \eps)$ be the answer.\\
		
		Given the sequence $((q^1,v^1),\cdots,(q^{t-1},v^{t-1}))$, let 
		  $q^t\in \calB(\beta_{\min},\beta_{\max})$ be the $t$-th query of $\calA_2$.\\

		Let (when $t=1$, set $p^1_{\min}\gets \alpha_{\min}$ and $p^1_{\max}\gets \alpha_{\max}$)\label{line: linfty construct}
		\begin{align*}p^t_{\min,i}&\leftarrow \max\left(\alpha_{\min,i},\max_{r\in[t-1]}(y^{(r,*)}_i-\norm{q^t-q^r}_{\infty})\right) \text{ for all } i\in[a]\quad\text{and}\\[0.5ex]
		p^t_{\max,i}&\leftarrow \min\left(\alpha_{\max,i},\min_{r\in[t-1]}(y^{(r,*)}_i+\norm{q^t-q^r}_{\infty})\right) \text{ for all } i\in[a].
		\end{align*}\label{define pr}\\
		
			Define a new function $g: \calB(p^t_{\min},p^t_{\max}) \to \calB(p^t_{\min},p^t_{\max},\pm \eps)$ as follows:
			$$ g(x)=\Big(f(x,q^t)_1,\cdots,f(x,q^t)_a\Big),
			\quad \text{for every $x\in \calB(p^t_{\min},p^t_{\max})$.}$$\\
			Run algorithm $\calA_1$ 
			  to find a solution 
			  $y^{(t,*)}\in\calB(p^t_{\min},p^t_{\max})$ to $\NonExp^\dagger_{\infty}(\eps,a)$
			  on $g$.

		Construct $v^t\in [0,1]^b$ as the query result to $q^t$: $$v^t_{j-a}=f(y^{(t,*)},q^t)_j,\quad\text{for each $j\in[a+b]\setminus [a]$}.$$\\ 
		
		If $\norm{v^t-q^t}_{\infty}\leq \eps$, 
		  \Return $(y^{(t,*)},q^t)$.
	}
	
\end{algorithm}
\begin{proof}[Proof of \Cref{thm: linfty decomposition}]
	Suppose that $\calA_1$ is an algorithm that solves $\NonExp^\dagger_{\infty}(\eps,a)$ in $q(\eps,a)$ queries and $t(\eps,a)$ time, and $\calA_2$ is an algorithm that solves $\NonExp^\dagger_{\infty}(\eps,b)$ in $q(\eps,b)$ queries and $t(\eps,b)$ time. We will show that \autoref{alg: linfty} can solve $\NonExp^\dagger_{\infty}(\eps,a+b)$ in $q(\eps,a)\cdot q(\eps,b)$ queries and $O((q(\eps,b)+t(\eps,a))\cdot t(\eps,b))$ time.\medskip

\def\hhh{h}	
	
\noindent \textbf{Overview.}
At a high level, \autoref{alg: linfty} 
  runs $\calA_2$ for $\NonExp^\dagger_{\infty}(\eps,b)$ on a function $\hhh: \calB(\beta_{\min},\beta_{\max})\to \calB(\beta_{\min},\beta_{\max},\pm \eps)$ built on the go using the input of the nonexpansive $$f: \calB\big((\alpha_{\min},\beta_{\min}),(\alpha_{\max},\beta_{\max})\big)\to \calB\big((\alpha_{\min},\beta_{\min}),(\alpha_{\max},\beta_{\max}),\pm \eps\big).$$
Let $q^1,\ldots,q^{t-1}\in \calB(\beta_{\min},\beta_{\max})$ be the $t-1$ queries that 
  $\calA_2$ has made so far, for some $t\ge 1$, and 
  let $v^1,\ldots,v^{t-1}\in \calB(\beta_{\min},\beta_{\max},\pm \eps)$ be the query results
  on $\hhh$.
Let $q^t\in \calB(\beta_{\min},\beta_{\max})$ be the new query made by $\calA_2$ in the $t$-th round.
We will use $f$ and the past queries and their associated results $((q^1,v^1),\cdots,(q^{t-1},v^{t-1}))$ to come up with a $v^t\in \calB(\beta_{\min},\beta_{\max},\pm \eps)$ 
  as the answer $h(q^t)$ to the query
  such that
\begin{flushleft}\begin{enumerate} \item
\Cref{lemma: linfty consistent}: All results $(q^1,v^1),\ldots,(q^t,v^t)$ 
  are consistent with the nonexpansion property, i.e., $\norm{v^r-v^{r'}}_{\infty}\leq \norm{q^r-q^{r'}}_{\infty}$ for all $r,r'\in[t]$; and 
  \item \Cref{lemma: linfty find a solution}: When $q^t$ is a solution to $\NonExp^\dagger_{\infty}(\eps,b)$ on $h$,
  i.e., $\norm{v^t-q^t}_{\infty}\leq \eps$,
  we obtain a solution to $\NonExp^\dagger_{\infty}(\eps,a+b)$ on 
  the original input function $f$.
\end{enumerate}\end{flushleft}

Note that the pairwise consistency of the answers (\Cref{lemma: linfty find a solution}) implies that there exists a nonexpansive function $h$ such that $h(q^r) = v^r$ for all $r \in [t]$, by Lemma 11 of \cite{CLY25}.
To obtain the answer $v^t$ to the query $q^t$, intuitively, we fix the last $b$ coordinates to be $q^t$ and find an $\eps$-fixed point of $f$ only for the first $a$ coordinates. We achieve so by running $\calA_1$ on a function $g:\calB(p^t_{\min},p^t_{\max})\to \calB(p^t_{\min},p^t_{\max},\pm \eps)$, where $\alpha_{\min}\preceq p^t_{\min}\preceq p^t_{\max}\preceq \alpha_{\max}$. The function $g$ is defined as the projection of $f$ onto the slice where the last $b$ coordinates are fixed to $q^t$.
The most crucial component in this procedure
  is to initialize the search space $\calB(p^t_{\min},p^t_{\max})$, rather than using the full domain $\calB(\alpha_{\min},\alpha_{\max})$, using 
  pairs $(q^r,v^r)$, $r\in [t-1]$, from previous rounds.\medskip

\noindent	\textbf{Correctness.} 
	We prove a sequence of lemmas about \autoref{alg: linfty}:

\begin{lemma}\label{hehehehe}
In each round $t$, 
we have 
\begin{enumerate}
	\item $\alpha_{\min}\preceq p^t_{\min}\preceq p^t_{\max}\preceq \alpha_{\max}$;
	\item $g$ is a nonexpansive map with respect to the $\ell_{\infty}$-norm; and
	\item $g(x)\in \calB(p^t_{\min},p^t_{\max},\pm \eps)$ for all $x\in \calB(p^t_{\min},p^t_{\max})$;
	\item $\norm{f(y^{(t,*)},q^t)_{[a]}-y^{(t,*)}}_{\infty}\leq \eps$.
\end{enumerate}
\end{lemma}

\begin{proof} 
We prove the lemma by induction on $t$. The base case when $t=1$ is trivial given that $p_{\min}^1 $ $=\alpha_{\min}$ and $p_{\max}^1=\alpha_{\max}$. For the induction step, we prove that the four items above hold for $t>1$. We assume by the inductive hypothesis that for any $r\in[t-1]$, we have (1) $\alpha_{\min}\preceq p^r_{\min}\preceq p^r_{\max}\preceq \alpha_{\max}$ and (4) $\norm{f(y^{(r,*)},q^r)_{[a]}-y^{(r,*)}}_{\infty}\leq \eps$. We show each item in order.\medskip

\noindent\textbf{First item.} $\alpha_{\min}\preceq p_{\min}^t$ and $p_{\max}^t\preceq \alpha_{\max}$ follow directly from the definition. To show $p_{\min}^t\preceq p_{\max}^t$, it suffices to show that
\[y^{(r,*)}_i-\norm{q^t-q^r}_{\infty}\leq y^{(r',*)}_i+\norm{q^t-q^{r'}}_{\infty}\] for all $i\in[a]$ and $r,r'\in[t-1]$. Equivalently, we want to show 
\[\norm{y^{(r,*)}-y^{(r',*)}}_{\infty}\leq \norm{q^t-q^{r'}}_{\infty}+\norm{q^t-q^r}_{\infty}\] for all $r<r'\in[t-1]$. Fix arbitrary $r<r'\in[t-1]$. By the inductive hypothesis we have that $p_{\min}^r\preceq p_{\max}^r$. Then we know $y^{(r',*)}\in \calB(p_{\min}^r,p_{\max}^r)$. This implies that \[y^{(r,*)}_i-\norm{q^{r'}-q^r}_{\infty}\leq y^{(r',*)}_i\leq y^{(r,*)}_i+\norm{q^{r'}-q^r}_{\infty}\] for all $i\in[a]$. So we know that $\norm{y^{(r,*)}-y^{(r',*)}}_{\infty}\leq \norm{q^{r'}-q^r}_{\infty}\leq \norm{q^{t}-q^r}_{\infty}+\norm{q^{t}-q^{r'}}_{\infty}.$

\vspace{0.2cm}
\noindent\textbf{Second item.} For the second item, notice that 
$$
\norm{g(x)-g(z)}_{\infty}=\norm{f(x,q^t)-f(z,q^t)}_{\infty}\leq \norm{x-z}_{\infty}.
$$
where the last inequality is because $f$ is a nonexpansive function.

\vspace{0.2cm}
\noindent\textbf{Third item.} For the third item, we first note an equivalent definiton of $\calB(p^t_{\min},p^t_{\max})$, i.e.,  
\begin{equation}\label{equation: Bpminpmax}
\calB(p^t_{\min},p^t_{\max})=\left(\bigcap_{r\in[t-1]}\calB(y^{(r,*)},y^{(r,*)},\pm \norm{q^t-q^r}_{\infty})\right)\cap \calB(\alpha_{\min},\alpha_{\max}).
\end{equation}

Similarly, an equivalent definiton of $\calB(p^t_{\min},p^t_{\max},\pm \eps)$ is as follows:
\begin{equation}\label{equation: Bpminpmax eps}
\calB(p^t_{\min},p^t_{\max},\pm \eps)=\left(\bigcap_{r\in[t-1]}\calB(y^{(r,*)},y^{(r,*)},\pm (\norm{q^t-q^r}_{\infty}+\eps))\right)\cap \calB(\alpha_{\min},\alpha_{\max},\pm \eps).
\end{equation}
Recall that $f: \calB((\alpha_{\min},\beta_{\min}),(\alpha_{\max},\beta_{\max}))\to \calB((\alpha_{\min},\beta_{\min}),(\alpha_{\max},\beta_{\max}),\pm \eps)$, so we have $g(x)\in \calB(\alpha_{\min},\alpha_{\max},\pm\eps)$.

Now fix any $r\in[t-1]$ and any $x\in \calB(y^{(r,*)},y^{(r,*)},\pm \norm{q^t-q^r}_{\infty})$. We will show that 
$$
g(x)\in \calB(y^{(r,*)},y^{(r,*)},\pm (\norm{q^t-q^r}_{\infty}+\eps)).
$$

First, note that
$$
\norm{g(x)-y^{(r,*)}}_{\infty}=\norm{f(x,q^t)_{[a]}-y^{(r,*)}}_{\infty}\leq \norm{f(x,q^t)_{[a]}-f(y^{(r,*)},q^r)_{[a]}}_{\infty}+\norm{f(y^{(r,*)},q^r)_{[a]}-y^{(r,*)}}_{\infty}.
$$

Since $f$ is nonexpansive, we know that $\norm{f(x,q^t)_{[a]}-f(y^{(r,*)},q^r)_{[a]}}_{\infty}\leq \norm{(x,q^t)-(y^{(r,*)},q^r)}_{\infty}$.
Since $x\in\calB(y^{(r,*)},y^{(r,*)},\pm \norm{q^t-q^r}_{\infty})$, we know that $\norm{(x,q^t)-(y^{(r,*)},q^r)}_{\infty}\leq \norm{x-y^{(r,*)}}_{\infty}$. Moreover, since $r<t$, we have $\norm{f(y^{(r,*)},q^r)_{[a]}-y^{(r,*)}}_{\infty}\leq \eps$ by the inductive hypothesis and thus,
$$
\norm{g(x)-y^{(r,*)}}_{\infty}\leq \norm{x-y^{(r,*)}}_{\infty}+\eps, \quad\text{which implies}\quad g(x)\in \calB(y^{(r,*)},y^{(r,*)},\pm (\norm{q^t-q^r}_{\infty}+\eps)).
$$

This finishes the proof of the third item, according to \Cref{equation: Bpminpmax} and \Cref{equation: Bpminpmax eps}.

\vspace{0.2cm}
\noindent\textbf{Fourth item.} Given that $g(x):\calB(p^t_{\min},p^t_{\max})\to \calB(p^t_{\min},p^t_{\max},\pm \eps)$, we know that the point $y^{(t,*)}$ that the algorithm $\calA_1$ outputs is an $\eps$-fixed point. Thus $\norm{g(y^{(t,*)})-y^{(t,*)}}_{\infty}\leq \eps$, which means $\norm{f(y^{(t,*)},q^t)_{[a]}-y^{(t,*)}}_{\infty}\leq \eps$.
\end{proof}

We are now ready to prove the two lemmas needed for the correctness of \autoref{alg: linfty}:

\begin{lemma}\label{lemma: linfty consistent}
	For every round $t\geq 1$ and every $r <t$, we have  $\norm{v^r-v^t}_{\infty}\leq \norm{q^r-q^t}_{\infty}$.
\end{lemma}
\begin{proof}
Since $v^t$ is defined by $f$ so clearly $v^t\in \calB(\beta_{\min},\beta_{\max},\pm \eps)$. Notice that
$$
\norm{v^t-v^r}_{\infty}=\max_{j\in[a+b]\setminus[a]} |f(y^{(t,*)},q^t)_j-f(y^{(r,*)},q^r)_j|\leq \norm{(y^{(t,*)},q^t) - (y^{(r,*)},q^r)}_{\infty}\leq \norm{q^t-q^r}_{\infty},
$$
where the last inequality follows from $y^{(t,*)}\in\calB(y^{(r,*)},y^{(r,*)},\pm \norm{q^t-q^r}_{\infty})$.
\end{proof}

\begin{lemma}\label{lemma: linfty find a solution}
	For every round $t$, if $\norm{v^t-q^t}_{\infty}\leq \eps$, then $\norm{f(y^{(t,*)},q^t)-(y^{(t,*)},q^t)}_{\infty}\leq \eps$.
\end{lemma}
\begin{proof}
	By the correctness of the algorithm $\calA_1$, we know that 
	$$
	\max_{j\in[a]} |f(y^{(t,*)},q^t)_j-(y^{(t,*)},q^t)_j|\leq \eps.
	$$
	If $\norm{v^t-q^t}_{\infty}\leq \eps$, then we know that 
	$$
	\max_{j\in[a+b]\setminus[a]} |f(y^{(t,*)},q^t)_j-(y^{(t,*)},q^t)_j|\leq \eps.
	$$
	Then the lemma follows.
\end{proof}

\noindent\textbf{Query complexity and time complexity.} For each round of \autoref{alg: linfty},  the algorithm $\calA_1$ is called once, which uses $q(\eps,a)$ queries and $t(\eps,a)$ time. The outer algorithm $\calA_2$ has no more than $q(\eps,b)$ rounds,
	which means the query complexity of \autoref{alg: linfty} is at most $q(\eps,a) \cdot q(\eps,b)$. In each round, line~\ref{line: linfty construct} takes $O(q(\eps,b))$ time, so each round of \autoref{alg: linfty} needs $O(q(\eps,b)+t(\eps,a))$ time. Overall the time complexity is $O((q(\eps,b)+t(\eps,a))\cdot t(\eps,b))$.
\end{proof}

\section{Algorithm for $\NonExp_1(\eps,k)$}
In this section we give an algorithm, \autoref{algorithm: time}, for $\Contraction_1(\eps,\gamma,k)$ with time complexity $O(k\log^{\lceil k/2\rceil}(1/\eps\gamma))$. \autoref{theorem: l1}  then follows from Observation~\ref{observation: l1 nonexp to contraction} below:
\begin{observation}\label{observation: l1 nonexp to contraction}
Let $f:[0,1]^k\mapsto [0,1]^k$ be a nonexpansive map under the $\ell_1$-norm.
Consider the map $g:[0,1]^k\mapsto[0,1]^k$ defined as $g(x):=(1-\eps/(2k))f(x)$. Clearly $g$ is a $(1-\eps/(2k))$-contraction with respect to the $\ell_1$-norm. Let $x\in [0,1]^k$ be any point with $\norm{g(x)-x}_1\leq \eps/2$. We have 
$$
\eps/2\ge \norm{g(x)-x}_1 =\norm{(1-\eps/(2k))f(x)-x}_1\ge \norm{f(x)-x}_1-(\eps/2k)\cdot\norm{f(x)}_1\geq \norm{f(x)-x}_1-\eps/2.
$$
This gives a black-box reduction 
  from $\NonExp_1(\eps,k)$ to $\Contraction_1(\eps/2,\eps/(2k),k)$, which is both query-efficient and time-efficient.
\end{observation}

Without loss of generality, we assume that $k$ is even; if not, we can easily pad with one extra dimension, setting $f(x)_{k+1}=0$ for all $x\in[0,1]^{k+1}$. We start with an overview of \autoref{algorithm: time}.

\begin{algorithm}[t!]
\caption{Recursive Algorithm for $\Contraction_1(\eps,\gamma,k)$}
\label{algorithm: time}
\KwIn{An even number $k$ and oracle access to a $(1-\gamma)$-contraction $f:[0,1]^k\mapsto[0,1]^k$ with respect to $\ell_1$-norm.}
\KwOut{A point $x\in[0,1]^k$ such that $\norm{f(x)-x}_1\leq \eps$.}

	\BlankLine
    Let $p_{\min}=(0,0)$ and $p_{\max}=(1,1)$.
    
    \For{$t=1,2,\ldots$}{
		Let $p_{\midd}=(p_{\min}+p_{\max})/2$.

		\If{$k>2$}{
		Define a new function $\phi:[0,1]^{k-2}\mapsto[0,1]^{k-2}$ as follows:
		$$\phi(y)\coloneqq \big(f(p_{\midd},y)_3,\dots,f(p_{\midd},y)_k\big).$$
		
		Recursively find a point $y^*\in [0,1]^{k-2}$ such that $\norm{\phi(y^*)-y^*}_1\leq \eps\gamma/4$.
		}
		
		Let $x=(p_{\midd},y^*)$ and \textbf{query} $f(x)$.\tcp{When $k=2$, $y^*$ is defined as empty.}
		
		\If{$\norm{f(x)-x}_1\leq \eps$}{
			\textbf{return} $x$ as an $\eps$-fixed point of $f$. \label{returnnnnnnn}%
			}
		
		Let $i\in\set{1,2}$ be such that $|f(x)_i-x_i|\geq |f(x)_{3-i}-x_{3-i}|$ (breaking ties arbitrarily). 
		
		\tcp{Find a (weakly) dominating coordinate.}
		\If{$f(x)_i>x_i$}{
				Let $p_{\min,i}=x_i+\eps\gamma/8$.
			}
		\Else{
				 Let $p_{\max,i}=x_i-\eps\gamma/8$.
			}
		}
\end{algorithm}

    \vspace{0.2cm}
    \noindent\textbf{Overview of \autoref{algorithm: time}.} Let $x^*=\Fix(f)$ be the unknown fixed point of $f$. The algorithm maintains two variables $p_{\min},p_{\max}\in [0,1]^2$ for the first two coordinates, such that during its execution, it always holds that $p_{\min,i}\leq x^*_i \leq p_{\max,i}$ for both $i\in\set{1,2}$. In other words, the projection of $x^*$ onto the first two coordinates lies within the rectangle defined by $p_{\min}$ and $p_{\max}$. Initially $p_{\min}$ is set to be $(0,0)$ and $p_{\max}$ is set to be $(1,1)$ so this holds trivially. 
    
    For each round, the algorithm fixes the first two coordinates as $p_{\midd}=(p_{\min}+p_{\max})/2$, and recursively finds an approximate fixed point $y^*$ for the remaining $k-2$ coordinates, with an improved approximation guarantee of $\eps\gamma/4$ instead of $\eps$. We then let $x=(p_{\midd},y^*)$.

    The algorithm then proceeds by querying $f(x)$. If $x$ turns out to be an $\eps$-fixed point then we are done. Otherwise, we know that
    $$\norm{f(x)-x}_1>\eps\quad\text{and}\quad\sum_{i=3}^k |f(x)_i-x_i|\leq \eps\gamma/4.$$ Thus, intuitively, most of the discrepancy between $f(x)$ and $x$ arises from the first two coordinates. We then select $i\in\set{1,2}$ such that $|f(x)_i-x_i|\geq |f(x)_{3-i}-x_{3-i}|$. Then we know that $$|f(x)_i-x_i|\geq \sum_{j\neq i}|f(x)_j-x_j|-\eps\gamma/4,$$ which means the discrepancy on $i$ ``almost'' dominates the discrepancy on all other coordinates. This will allow us to shrink $p_{\max,i}-p_{\min,i}$ by half (in fact, slightly more than half).

    We now formalize this with the following technical lemma.
    \begin{lemma}\label{lemma: timexxx}
        Let $k\geq 2$. Let $f:[0,1]^k\mapsto [0,1]^k$ be a $(1-\gamma)$-contraction map, and let $x^*=\emph{\Fix}(f)$. If a point $x\in[0,1]^k$ satisfies $\norm{f(x)-x}_1>\eps$ but $\sum_{j=3}^k|f(x)_j-x_j|\leq \eps\gamma/4$, then we have 
        \begin{itemize}
            \item $x_i^*>x_i+\eps\gamma/8$ if $f(x)_i>x_i$;
            \item $x_i^*<x_i-\eps\gamma/8$ if $f(x)_i<x_i$,
        \end{itemize}
        where $i\in\set{1,2}$ is such that $|f(x)_i-x_i|\geq |f(x)_{3-i}-x_{3-i}|$.
    \end{lemma}
    \begin{proof}
        Without loss of generality, assume that $i=1$, namely, $|f(x)_1-x_1|\geq |f(x)_2-x_2|$.  To further simplify notations, we also assume $f(x)_1>x_1$, and $f(x)_2\geq x_2$. Then our goal is to show $x^*_1>x_1+\eps\gamma/8$, and the lemma follows from symmetry. 

        Given that $f$ is a $(1-\gamma)$-contraction and $x^*$ is the fixed point of $f$, we have that 
        \begin{equation}\label{eqhhhhh}
            \norm{f(x)-x^*}_1=\norm{f(x)-f(x^*)}_1\leq (1-\gamma)\cdot \norm{x-x^*}_1.
        \end{equation}
        Furthermore, we note that $\norm{x-x^*}_1>\eps/2$, as otherwise, we have 
        $$\norm{f(x)-x}_1\leq \norm{f(x)-x^*}_1+\norm{x^*-x}_1=\norm{f(x)-f(x^*)}_1+\norm{x^*-x}_1\leq 2\norm{x^*-x}_1\leq \eps.$$
        
        \paragraph{Proof by contradiction.} For the sake of contradictions, we assume that $x^*_1\leq x_1+\eps\gamma/8$. From this, we will give a non-trivial upper bound on $\norm{x-x^*}_1-\norm{f(x)-x^*}_1$ below: \begin{equation}\label{hehe1}
        \norm{x-x^*}_1-\norm{f(x)-x^*}_1\leq \eps\gamma/2.
        \end{equation} 
However, from (\ref{eqhhhhh}), we know 
        $$\gamma\norm{x-x^*}_1\leq \norm{x-x^*}_1-\norm{f(x)-x^*}_1\leq \eps\gamma/2,$$
        which implies $\norm{x-x^*}_1\leq \eps/2$, leading to a contradiction.

        To prove \Cref{hehe1}, we partition $\norm{x-x^*}_1-\norm{f(x)-x^*}_1$ into two parts: $$A=\sum_{j=1,2}\left(|x_j-x^*_j|-|f(x)_j-x^*_j|\right)\quad\text{and}\quad B=\sum_{j=3}^k\left(|x_j-x^*_j|-|f(x)_j-x^*_j|\right).$$
        First, it is clear that by triangle inequality and assumption of the lemma, $$B\leq \sum_{j=3}^k|x_j-f(x)_j|\leq \eps\gamma/4.$$ 
        To bound $A$, we note that by assumption that $\norm{f(x)-x}_1>\eps$ but $\sum_{j=3}^k|f(x)_j-x_j|\leq \eps\gamma/4$, we have $f(x)_1-x_1>(\eps-\eps\gamma/4)/2$. This implies that
        $$x^*_1\leq x_1+\eps\gamma/8<x_1+(\eps-\eps\gamma/4)/2<f(x)_1.$$
        Note also that since $x^*_1\leq x_1+\eps\gamma/8$, we have $|x_1-x^*_1|\leq (x_1-x^*_1)+\eps\gamma/4$.
        Thus 
        \begin{align*}
            A &= |x_1 - x^*_1| - |f(x)_1 - x^*_1| + |x_2 - x^*_2| - |f(x)_2 - x^*_2| \\
              &= |x_1 - x^*_1| - (f(x)_1 - x^*_1) + |x_2 - x^*_2| - |f(x)_2 - x^*_2| \\
              &\leq \eps\gamma/4 + (x_1 - x^*_1) - (f(x)_1 - x^*_1) + |x_2 - x^*_2| - |f(x)_2 - x^*_2| \\
              &= \eps\gamma/4 + (x_1 - f(x)_1) + |x_2 - x^*_2| - |f(x)_2 - x^*_2| \\
              &\leq \eps\gamma/4 + x_1 - f(x)_1+ f(x)_2-x_2\\
              &\leq \eps\gamma/4,
        \end{align*}
        where the last inequality follows from $|f(x)_1-x_1|\geq |f(x)_2-x_2|$. This finishes the proof that $\norm{x-x^*}_1-\norm{f(x)-x^*}_1\leq \eps\gamma/2$, and the proof of the lemma as well.
    \end{proof}

    Given \Cref{lemma: timexxx}, we know that at the end of each round, we have $p_{\min,i}\leq x^*_i \leq p_{\max,i}$ for both $i\in\set{1,2}$.
    Furthermore, if $\|p_{\max}-p_{\min}\|_\infty<\eps\gamma/8$ at the beginning of a round, then $x$ on line~\ref{returnnnnnnn} must be an $\eps$-fixed point of $f$ during that round.
    Since initially $p_{\max,i}-p_{\min,i}=1$, we know that the algorithm must terminate within $O(\log(1/\eps\gamma))$ rounds to find an $\eps$-fixed point.
    
    The recursive expression for the time complexity of the algorithm is $$T(\eps,\gamma,k)=O(\log(1/\eps\gamma))\cdot \big(O(1)+T(\eps\gamma,\gamma,k-2)\big).$$
    Solving this recurence shows that $T(\eps,\gamma,k)=O(k\log^{\lceil k/2\rceil}(1/\eps\gamma))$.

\section{Discussion and Future Directions}

In this paper we gave $O(\log^{\lceil k/2\rceil}(1/\eps))$-time algorithms for 
  $\NonExp_\infty(\eps,k)$ and $\NonExp_1(\eps,k)$.~{A natural next step is to give an $O(\log(1/\eps))$-time algorithm for $\NonExp_{\infty}(\eps,3)$ or $\NonExp_1(\eps,3)$. Such an algorithm for $\NonExp_{\infty}(\eps,3)$ in particular will imply an $O(\log^{\lceil k/3\rceil}(1/\eps))$-time algorithm for general $k$ by \Cref{thm: linfty decomposition}.} However, there are some challenges. {At a high level, the $O(\log(1/\eps))$-time algorithms for both $\NonExp_{\infty}(\eps,2)$ and $\NonExp_{1}(\eps,2)$ are due to the crucial observation that the search space can always be maintained convex. This is no longer true when $k=3$ for both problems, even after one query (see figures in \cite{CLY25} for the case with the $\ell_\infty$-norm).} 

Another direction is to understand trade-offs between time and query complexities. Note that there are highly query-efficient algorithms~\cite{CLY25,haslebacher2025query} for contraction fixed points; \mbox{however,} known implementions of these query-efficient algorithms all have time complexity $\Omega((1/\eps)^k)$ due to their brute-force nature. On the other hand, when focusing on time-efficient algorithms such~as those of \cite{SS02,SS03,FGMS20} as well as the present paper, their query complexities are essentially the same as the time complexities. Can we achieve some intermediate results? For example, concretely, can one give an algorithm for either $\NonExp_{\infty}(\eps,3)$ or $\NonExp_{1}(\eps,3)$ with time complexity $O(\log^2(1/\eps))$ but only uses $O(\log(1/\eps))$ many queries?

\begin{flushleft}
\bibliographystyle{alpha}
\bibliography{ref}

\newcommand{\etalchar}[1]{$^{#1}$}
\begin{thebibliography}{DGK{\etalchar{+}}17}

\bibitem[Ban22]{Ban1922}
Stefan Banach.
\newblock Sur les op{\'e}rations dans les ensembles abstraits et leur application aux {\'e}quations int{\'e}grales.
\newblock {\em Fundamenta mathematicae}, 3(1):133--181, 1922.

\bibitem[Bel57]{bellman1957markovian}
Richard Bellman.
\newblock A {M}arkovian decision process.
\newblock {\em Journal of mathematics and mechanics}, pages 679--684, 1957.

\bibitem[BFG{\etalchar{+}}25]{batzioumonotonecontraction}
Eleni Batziou, John Fearnley, Spencer Gordon, Ruta Mehta, and Rahul Savani.
\newblock Monotone contractions.
\newblock In {\em Proceedings of the 57th Annual ACM Symposium on Theory of Computing}, pages 507--517, 2025.

\bibitem[CJK{\etalchar{+}}22]{CJKLS22}
Cristian~S. Calude, Sanjay Jain, Bakhadyr Khoussainov, Wei Li, and Frank Stephan.
\newblock Deciding parity games in quasi-polynomial time.
\newblock {\em {SIAM} J. Comput.}, 51(2):17--152, 2022.

\bibitem[CL55]{CL55}
E.~A. Coddington and N.~Levinson.
\newblock {\em Theory of Ordinary Differential Equations}.
\newblock McGraw Hill, 1955.

\bibitem[CL22]{CL22}
Xi~Chen and Yuhao Li.
\newblock {Improved Upper Bounds for Finding Tarski Fixed Points}.
\newblock In {\em {EC}}, pages 1108--1118. {ACM}, 2022.

\bibitem[CLY25]{CLY25}
Xi~Chen, Yuhao Li, and Mihalis Yannakakis.
\newblock Computing a fixed point of contraction maps in polynomial queries.
\newblock {\em J. ACM}, 72(4), July 2025.

\bibitem[Con92]{Con92}
Anne Condon.
\newblock The complexity of stochastic games.
\newblock {\em Information and Computation}, 96(2):203--224, 1992.

\bibitem[Den67]{De67}
Eic~V. Denardo.
\newblock Contraction mappings underlying dynamic programming.
\newblock {\em SIAM Review}, 9(2):165--177, 1967.

\bibitem[DGK{\etalchar{+}}17]{dohrau2017arrival}
J{\'e}r{\^o}me Dohrau, Bernd G{\"a}rtner, Manuel Kohler, Ji{\v{r}}{\'\i} Matou{\v{s}}ek, and Emo Welzl.
\newblock {ARRIVAL: a zero-player graph game in NP $\cap$ coNP}.
\newblock In {\em A Journey Through Discrete Mathematics: A Tribute to Ji{\v{r}}{\'\i} Matou{\v{s}}ek}, pages 367--374. Springer, 2017.

\bibitem[DP11]{DP11}
Constantinos Daskalakis and Christos~H. Papadimitriou.
\newblock Continuous local search.
\newblock In {\em Proceedings of the Twenty-Second Annual {ACM-SIAM} Symposium on Discrete Algorithms, {SODA} 2011}, pages 790--804. {SIAM}, 2011.

\bibitem[DTZ18]{DTZ18}
Constantinos Daskalakis, Christos Tzamos, and Manolis Zampetakis.
\newblock A converse to {B}anach's fixed point theorem and its {CLS}-completeness.
\newblock In Ilias Diakonikolas, David Kempe, and Monika Henzinger, editors, {\em Proceedings of the 50th Annual {ACM} {SIGACT} Symposium on Theory of Computing, {STOC} 2018, Los Angeles, CA, USA, June 25-29, 2018}, pages 44--50. {ACM}, 2018.

\bibitem[EJ91]{EJ91}
E.~A. Emerson and C.~Jutla.
\newblock Tree automata, $\mu$-calculus and determinacy.
\newblock In {\em Proceedings IEEE Symp. on Foundations of Computer Science}, pages 368--377, 1991.

\bibitem[EM79]{ehrenfeucht1979positional}
Andrzej Ehrenfeucht and Jan Mycielski.
\newblock Positional strategies for mean payoff games.
\newblock {\em International Journal of Game Theory}, 8(2):109--113, 1979.

\bibitem[EY10]{EY10}
Kousha Etessami and Mihalis Yannakakis.
\newblock On the complexity of {N}ash equilibria and other fixed points.
\newblock {\em {SIAM} J. Comput.}, 39(6):2531--2597, 2010.

\bibitem[FGHS23]{FGHS23}
John Fearnley, Paul Goldberg, Alexandros Hollender, and Rahul Savani.
\newblock The complexity of gradient descent: {CLS} = {PPAD} {\(\cap\)} {PLS}.
\newblock {\em J. {ACM}}, 70(1):7:1--7:74, 2023.

\bibitem[FGMS20]{FGMS20}
John Fearnley, Spencer Gordon, Ruta Mehta, and Rahul Savani.
\newblock Unique end of potential line.
\newblock {\em J. Comput. Syst. Sci.}, 114:1--35, 2020.

\bibitem[FPS22]{FPS22}
John Fearnley, D{\"o}m{\"o}t{\"o}r P{\'a}lv{\"o}lgyi, and Rahul Savani.
\newblock A faster algorithm for finding {T}arski fixed points.
\newblock {\em ACM Transactions on Algorithms (TALG)}, 18(3):1--23, 2022.

\bibitem[GHH{\etalchar{+}}18]{gartner2018arrival}
Bernd G{\"a}rtner, Thomas~Dueholm Hansen, Pavel Hub{\'a}cek, Karel Kr{\'a}l, Hagar Mosaad, and Veronika Sl{\'\i}vov{\'a}.
\newblock {ARRIVAL: Next Stop in CLS}.
\newblock In {\em 45th International Colloquium on Automata, Languages, and Programming (ICALP 2018)}, pages 60--1. Schloss Dagstuhl--Leibniz-Zentrum f{\"u}r Informatik, 2018.

\bibitem[GHH21]{gartner2021subexponential}
Bernd G{\"a}rtner, Sebastian Haslebacher, and Hung~P Hoang.
\newblock {A Subexponential Algorithm for ARRIVAL}.
\newblock In {\em 48th International Colloquium on Automata, Languages, and Programming (ICALP 2021)}, pages 69--1. Schloss Dagstuhl--Leibniz-Zentrum f{\"u}r Informatik, 2021.

\bibitem[G{\"u}n89]{gunther1989}
Matthias G{\"u}nther.
\newblock Zum einbettungssatz von {J. Nash}.
\newblock {\em Mathematische Nachrichten}, 144(1):165--187, 1989.

\bibitem[Has25]{haslebacher2025arrival}
Sebastian Haslebacher.
\newblock {ARRIVAL: Recursive Framework $\ell_1$-Contraction}.
\newblock In {\em 52nd International Colloquium on Automata, Languages, and Programming (ICALP 2025)}, pages 95--1. Schloss Dagstuhl--Leibniz-Zentrum f{\"u}r Informatik, 2025.

\bibitem[HKS99]{HKS99}
Z.~Huang, Leonid~G. Khachiyan, and Christopher~(Krzysztof) Sikorski.
\newblock Approximating fixed points of weakly contracting mappings.
\newblock {\em J. Complex.}, 15(2):200--213, 1999.

\bibitem[HLSW25]{haslebacher2025query}
Sebastian Haslebacher, Jonas Lill, Patrick Schnider, and Simon Weber.
\newblock Query-efficient fixpoints of $\ell_p$-contractions.
\newblock In {\em {FOCS}}, 2025.
\newblock To appear.

\bibitem[Hol21]{Hol21}
Alexandros Hollender.
\newblock {\em Structural results for total search complexity classes with applications to game theory and optimisation}.
\newblock PhD thesis, University of Oxford, {UK}, 2021.

\bibitem[How60]{howard1960dynamic}
Ronald~A Howard.
\newblock {\em Dynamic programming and {M}arkov processes.}
\newblock John Wiley, 1960.

\bibitem[Jur98]{jurdzinski1998deciding}
Marcin Jurdzi{\'n}ski.
\newblock {Deciding the winner in parity games is in UP $\cap$ co-UP}.
\newblock {\em Information Processing Letters}, 68(3):119--124, 1998.

\bibitem[Kar17]{karthik2017did}
CS~Karthik.
\newblock Did the train reach its destination: The complexity of finding a witness.
\newblock {\em Information Processing Letters}, 121:17--21, 2017.

\bibitem[MP91]{MegiddoPapadimitriou}
Nimrod Megiddo and Christos~H. Papadimitriou.
\newblock On total functions, existence theorems and computational complexity.
\newblock {\em Theor. Comput. Sci.}, 81(2):317--324, 1991.

\bibitem[Nas56]{nash1956imbedding}
John Nash.
\newblock The imbedding problem for {R}iemannian manifolds.
\newblock {\em Annals of mathematics}, 63(1):20--63, 1956.

\bibitem[Sha53]{Shapley53}
L.~Shapley.
\newblock Stochastic games.
\newblock {\em Proc. Nat. Acad. Sci.}, 39(10):1095--1100, 1953.

\bibitem[SLJ89]{stokey1989recursive}
Nancy~L Stokey and Robert~E Lucas~Jr.
\newblock {\em Recursive methods in economic dynamics}.
\newblock Harvard University Press, 1989.

\bibitem[SS02]{SS02}
Spencer~D. Shellman and Kris Sikorski.
\newblock A two-dimensional bisection envelope algorithm for fixed points.
\newblock {\em J. Complex.}, 18(2):641--659, 2002.

\bibitem[SS03]{SS03}
Spencer~D. Shellman and Christopher~(Krzysztof) Sikorski.
\newblock A recursive algorithm for the infinity-norm fixed point problem.
\newblock {\em J. Complex.}, 19(6):799--834, 2003.

\bibitem[STW93]{STW93}
Christopher~(Krzysztof) Sikorski, Chey{-}Woei Tsay, and Henryk Wozniakowski.
\newblock An ellipsoid algorithm for the computation of fixed points.
\newblock {\em J. Complex.}, 9(1):181--200, 1993.

\bibitem[ZP96]{ZP96}
Uri Zwick and Mike Paterson.
\newblock The complexity of mean payoff games on graphs.
\newblock {\em Theor. Comput. Sci.}, 158(1{\&}2):343--359, 1996.

\end{thebibliography}
\end{flushleft}

\end{document}